\documentclass{article}
\usepackage{amsmath,amssymb,latexsym,stmaryrd,amsthm,graphicx,color,framed}

\newtheorem{theorem}{\noindent Theorem}
\newtheorem{lemma}{\noindent Lemma}

\newtheorem{definition}{\noindent Definition}

\newcommand{\tuple}[1]{{\mbox{$\langle#1\rangle$}}}

\newcommand{\ignore}[1]{}

\def\unnamed#1#2{
\parbox{2.0in}{
\begin{tabbing}
\hspace{1em}\= #1 \\ \> \parbox{.75in}{\noindent \hrule ~} #2
\end{tabbing}
}}

\def\ant#1{\\ \> $#1$}

\newcommand{\intype}[2]{#1\!:\!#2}

\def\dcolon{\mathrel{:}\joinrel\mathrel{\mkern+5mu}\joinrel\mathrel{:}}
\newcommand{\inntype}[2]{#1\dcolon#2}
\newcommand{\bool}{\mathrm{\bf Bool}}

\newcommand{\sett}{\mathbf{Set}}

\newcommand{\class}{\mathbf{Class}}

\newcommand{\double}[1]{\left\llbracket #1 \right\rrbracket}

\def\unnamed#1#2{
\parbox{2.0in}{
\begin{tabbing}
\hspace{1em}\= #1 \\ \> \parbox{.75in}{\noindent \hrule ~} #2
\end{tabbing}
}}

\newcommand{\type}{\mathbf{Type}}

\newcommand{\true}{\mathbf{ True}} \newcommand{\false}{\mathbf{False}}

\newcommand{\bij}{\mathbf{Bij}}

\title{MathZero, the Classification Problem, and Set-Theoretic Type Theory}

\author{David McAllester\footnote{Toyota Technological Institute at Chicago}}

\date{}

\begin{document}
\maketitle{}

\begin{abstract}
  AlphaZero learns to play go, chess and shogi at a superhuman level
  through self play given only the rules of the game.  This raises the
  question of whether a similar thing could be done for mathematics ---
  a MathZero. MathZero would require a formal foundation and an objective.
  We propose the foundation of set-theoretic dependent
  type theory and an objective
  defined in terms of the classification problem --- the problem of classifying concept instances
  up to isomorphism.  The natural numbers arise as the solution to
  the classification problem for finite sets.  Here we generalize classical Bourbaki set-theoretic isomorphism
  to set-theoretic dependent type theory.  To our knowledge we give the first isomorphism inference rules for
  set-theoretic dependent type theory with propositional set-theoretic equality.
  The presentation is intended to be accessible to mathematicians with no prior exposure to type theory.
\end{abstract}
  
\section{Introduction}

AlphaZero learns to play go, chess and shogi at a superhuman level
through self play given only the rules of the game.  This raises the
question of whether a similar thing could be done for mathematics ---
a MathZero.  A necessary first step
is a definition of some form of ``mathematics game''. What are the rules
and what is the objective?  Presumably the rules correspond to some formal foundation
and the objective somehow involves formulating and proving theorems.
Here we propose a foundation of set-theoretic dependent type theory
and an objective based on the classification problem --- the problem
of enumerating the instances of a given concept up to isomorphism.

Isomorphism is central to the structure of mathematics.
Mathematics is organized around concepts such as graphs, groups, topological
spaces and manifolds each of which is associated with a notion of
isomorphism.  Each concept is associated with a classification problem
--- the problem of enumerating the instances of a given concept up to
isomorphism.  We also have the related notions of symmetry and
canonicality.  There is no canonical point on a geometric circle ---
any point can be mapped to any other point by rotating the circle.  A
rotation of the circle is an isomorphism of the circle with itself ---
a symmetry or automorphism.  Similarly, there is no canonical basis
for a finite dimensional vector space.  For any basis there is a
symmetry (automorphism) of the vector space which moves the basis to a
different basis.  Isomorphism is also central to understanding representation.
A group can be represented by a family of permutations.  Different
(non-isomorphic) families of permutation can represent the same group
(up to isomorphism).

People have a strong intuition that isomorphic objects are ``the same''.
We talk about ``the'' complete graph on five nodes or ``the''
topological sphere. A graph-theoretic
property must have the same truth value on any two isomorphic graphs.
This can be viewed as a tautology --- a property is
by definition graph-theoretic if it respects graph isomorphism.
But this tautological interpretation misses
the fact that being a graph theoretic property can be guaranteed by syntactic
well-formedness conditions.
Under an appropriate type system syntactically well-formed expressions automatically
respect isomorphism.

Martin-L\"{o}f dependent type theory \cite{mltt} is capable of expressing
statements and proofs of general mathematics \cite{COQ} while placing strong
syntactic well-formedness constraints on expressions. However, dependent type
theory is typically formulated in terms of propositions as types with
equality handled by axiom J.  This formulation can be challenging for
the uninitiated. Typically no meaning for the notation is specified with various interpretations (models)
being possible.

From the beginning it has been clear that dependent type theory can be
given a set-theoretic interpretation --- a set-theoretic model.  However,
this interpretation has generally been eschewed by type theorists.
Here we consider a version of set-theoretic type theory that avoids propositions as
types and axiom J and that gives transparent set-theoretic meanings for all
constructs. This semantically transparent type theory is adequate for formalizing general
mathematics.  We show that syntactic well-formedness constraints inherited from Martin-L\"{o}f type theory
can guarantee that all expressions respect isomorphism.

A definition of structures and isomorphism was given by Bourbaki.
As interpreted within type theory, Bourbaki defines a particular class of structures, such as the class of groups,
by a type expression.  Bourbaki-style structure types are restricted to have
all set variables declared at the top level of the expression. The top level set variables declare ``carrier sets''.
Most familiar structure classes, such as groups, have only a single carrier set --- the group elements.
In addition to the carrier sets, a structure type specifies functions and predicates over the carrier sets --- the structure
imposed on the carriers.
Two instances of a structure type are isomorphic if there exists a system of bijections
between their carrier sets which identifies their structure --- which carries the structure of the first to
the structure of the second. The classical definition of a structure
type specifies structure by a simple (non-dependent) type.  We generalize that here to handle dependent structure
types.

Section \ref{sec:rules} gives two inference rules for isomorphism.
The first is the congruence rule which states
that isomorphic objects are inter-substitutable in well-formed contexts.
The second is the structure rule which states
a generalization of Bourbaki isomorphism.
The main result of this paper is a semantic definition of isomorphism applying to all types (not just structure types)
and a proof of isomorphism congruence.  We also show that the structure rule and the congruence rule together
are complete for isomorphism --- they allow any isomorphism relationship to be expressed directly
in the base language without linguistic extensions for isomorphism.

Isomorphism was handled at an abstract level in the groupoid model of type theory \cite{GRPD}.
However, a purely set-theoretic treatment with transparent set-theoretic semantics and set-theoretic propositional equality
should make the formal treatment more accessible to mathematicians outside of the type theory community.
To our knowledge inference rules for isomorphism in this setting have not previously been formulated.

Section~\ref{sec:notation} gives a high level presentation of set-theoretic dependent type theory
including the inference rules for isomorphism.
Section~\ref{sec:semantics} more formally defines the semantics and well-formedness
constraints.
Section~\ref{sec:isomorphism} gives a semantic definition of isomorphism covering all
type expressions and proves isomorphism congruence.

\section{An Overview of Set-Theoretic Type Theory}
\label{sec:notation}

This section first lists the constructs of the type theory and
their set-theoretic meanings.  It then describes the well-formedness constraints
at a high level with the precise specification given in section~\ref{sec:semantics}.
It states the structure rule and the congruence rule
and defines cryptomorphism. Cryptomorphism is used to
show that the structure and congruence rules are complete for isomorphism.

\subsection{The Constructs}
\label{sec:meaning}

We consider a set-theoretic dependent type theory consisting of the following constructs.

\begin{itemize}
\item The constant symbol $\sett$ denotes the class of all sets.

\item The constant symbol $\bool$ denotes the type containing the two
  truth values.

\item A pair expression $\tuple{s,u}$ denotes the pair of $s$ and $u$
  and projections $\pi_1(e)$ and $\pi_2(e)$ denote the first and
  second component respectively of the pair $e$.
\end{itemize}

\begin{itemize}
\item A lambda expression $\lambda\;\intype{x}{\sigma}\; e[x]$ denotes
  the (set-theoretic) function mapping an element $x$ in the set $\sigma$ to
  the value $e[x]$.

\item An application $f(e)$ denotes the value of the function $f$ on
  argument $e$.
\end{itemize}

\begin{itemize}
\item The formula $e_1 = e_2$ is true if $e_1$ is the same as $e_2$
  (set-theoretic equality).

\item $\forall\intype{x}{\sigma}\;\Phi[x]$ is true if for every element $x$
  of the set or class $\sigma$ we have that $\Phi[x]$ is true.

\item The Boolean formulas $\neg \Phi$, $\Phi \vee \Psi$, $\Phi  \Rightarrow \Psi$
  and $\Phi \Leftrightarrow \Psi$ have their
  classical Boolean meaning.
\end{itemize}

\begin{itemize}
\item The dependent function type
  $\Pi_{\intype{x\;}{\;\sigma}}\;\tau[x]$ denotes the set of all
  set-theoretic functions $f$ with domain set $\sigma$ and such that
  for all $x \in \sigma$ we have that $f(x) \in \tau[x]$. If $x$ does
  not occur in $\tau$ we abbreviate
  $\Pi_{\intype{x\;}{\;\sigma}}\;\tau$ as $\sigma \rightarrow \tau$.

\item The dependent pair type $\sum_{\intype{x\;}{\;\sigma}}\;\tau[x]$ denotes the
  set or class of all pairs $\tuple{x,y}$ with $x \in  \sigma$ and $y \in \tau[x]$.
  If $x$ does not occur in $\tau$ then
  we abbreviate $\sum_{\intype{x\;}{\;\sigma}}\;\tau$ as $\sigma \times \tau$.

\item The subtype $S_{\intype{x\;}{\;\sigma}}\;\Phi[x]$ denotes the
  set or class of all $x \in \sigma$ such that $\Phi[x]$ is true.
\end{itemize}

We distinguish set-level expressions from class expressions both syntactically and semantically.
An expression is syntactically set-level
if occurrences of the constant $\sett$ are restricted to be inside formulas (Boolean expressions).
A class expression is one in which the constant $\sett$ occurs outside of a formula.
Set-level expressions denote
set-level values while class expressions denote classes which are taken to be of a different semantic kind from sets.
The set/class distinction is discussed in more detail in section~\ref{sec:universe}.

As an example, the class of all groups can be defined as
\begin{eqnarray*}
\mathrm{Magma} & \equiv & \sum_{\intype{s\;}{\;\sett}}\;s \times s
\rightarrow s \\ \\ \mathrm{Group} & \equiv &
S_{\intype{M\;}{\;\mathrm{Magma}}}\;\Phi(M)
\end{eqnarray*}
where $\Phi(M)$ states the group axioms.

Under the meanings defined above the same value can be in different types.  Most
significantly, a pair $\tuple{u,v}$ in the type $\sum_{\intype{x\;}{\;\sigma}}\;\tau[x]$
is also in the type $\sigma \times \tau[v]$.  The well-formedness conditions do not prevent this. A pair $\tuple{s,f}$ in the type Magma
is also in the very different type $\sett \times ((s \times s) \rightarrow s)$.

As another example consider $\sum_{\intype{s\;}{\;\sett}}\;\sum_{\intype{w\;}{\;\sett}}\;s\rightarrow w$
and  $\sum_{\intype{s\;}{\;\sett}}\;s\rightarrow w$ where $w$ is a free set variable in the second expression.
The first type allows
isomorphisms defined by two bijections --- bijections on both $s$ and $w$ --- while the second type only allows isomorphisms defined
by a single bijection on $s$.  The classification problem for the second type gives rise to the concept of a bag of (or multiset of) $w$.

It is important that functions are restricted to be set-level (as opposed to being functors on proper classes).
To see why consider a predicate $\intype{P}{\mathbf{Group} \rightarrow \bool}$.  In the
presence of such a predicate we have that $P(G)$ is a meaningful formula for any group $G$. In this
case we need that for two isomorphic groups $G$ and $G'$ we have $P(G) \Leftrightarrow P(G')$.  This
requires that the notion of isomorphism is built into the semantics of the type $\mathbf{Group}
\rightarrow \bool$.  This is done in the groupoid model \cite{GRPD}. Here we take a
different approach based on the observation that set-level function spaces require no such care ---
the fully naive definition of set-level function spaces just works.  Functors then arise naturally
as terms of the language.  Isomorphism congruence implies that functors defined by terms of the language automatically
respect isomorphism --- they are automatically functors between groupoids.

\subsection{Well-Formedness}

While the set-theoretic meanings of the constructs are transparent, the well-formedness conditions are more subtle.
Here we describe the well-formedness conditions at a high level and postpone a precise definition to section~\ref{sec:semantics}.

Well-formedness is relative to a context declaring the types for variables and
stating assumptions about those variables --- a
context $\Gamma$ consists of variable declarations
$\intype{x}{\tau}$ and assumptions $\Phi$.  The well-formedness conditions determine when an expression $e$ is well-formed under a context $\Gamma$.
Contexts themselves are subject to well-formedness constraints. For $\Gamma;\intype{x}{\sigma}$ to be well-formed
$\sigma$ must be a well-formed type expression under $\Gamma$ and for $\Gamma;\Phi$ to be well-formed $\Phi$ must be a well-formed Boolean expression under $\Gamma$.
All context are constructed in this way starting with the empty context $\epsilon$ and constants $\sett$ and $\bool$.
We write $\Gamma \models \intype{e}{\sigma}$ to mean that $\sigma$ is a well-formed type under $\Gamma$ and that $e$ is a well-formed expression under $\Gamma$ and
that for all variable interpretations satisfying the context $\Gamma$ we have that the value of $e$ is a member of the value of $\sigma$.
We write $\Gamma \models \Phi$ to mean that $\Phi$ is a well-formed Boolean expression under $\Gamma$ and that for every variable interpretation satisfying $\Gamma$
we have that $\Phi$ is true.

While the details of well-formedness constraints are postponed to section~\ref{sec:semantics} there are some cases that are particularly significant.
For a set-theoretic equality
$u = v$ to be well-formed we must first be able to derive $\Gamma \models \intype{u}{\sigma}$ and $\Gamma \models \intype{v}{\sigma}$
for some set expression $\sigma$.
Without this constraint we would have
$$\intype{s}{\sett};\;\intype{x}{s};\;\intype{S}{\left(\sum_{\intype{w\;}{\;\sett}}\;{w}\right)} \models \intype{(x=\pi_2(S))}{\bool}.$$
This would yield a counter-example to isomorphism congruence as isomorphic values of $S$ yield different truth values for the formula.

Similarly, for an application $f(e)$ to be well-formed under $\Gamma$ we must first be able to derive $\Gamma \models \intype{f}{(\prod_{\intype{x\;}{\;\sigma}}\;\tau[x])}$ and $\Gamma \models \intype{u}{\sigma}$
for some dependent function type $\prod_{\intype{x\;}{\;\sigma}}\;\tau[x]$.  This restriction is used in the proof of the bijectivity lemma stated in section~\ref{subsec:isomorphism}.

\subsection{Isomorphism Inference Rules}
\label{sec:rules}

We now extend the language with Boolean formulas of the form $u =_\sigma v$ which intuitively mean that $u$ and $v$ are isomorphic instances of the class $\sigma$.
A precise semantics for this formula for general class expressions is given in section~\ref{subsec:isomorphism}.
We give two inference rules for isomorphism --- the structure rule and the congruence rule.  The congruence rule is simpler and we state it first.

\vspace{-3ex}
~ \hfill\unnamed
    {
      \ant{\Gamma;\intype{x}{\sigma} \models \intype{e[x]}{\tau}\;\;\;\mbox{$x$ not free in $\tau$}}
      \ant{\Gamma \models u =_\sigma v}
    }{
      \ant{\Gamma \models e[u] =_\tau e[v]}
      }
\hfill ~

To state the structure rule we first define structure types.

\begin{definition}
  A structure type is either a set expression or a class expression of the form 
  $\sum_{\intype{s_1\;}{\;\sett}} \cdots \sum_{\intype{s_n\;}{\;\sett}}\;\tau$ where $\tau$ is a set expression.
\end{definition}

The traditional (Bourbaki) notion of structure type requires $\tau$ to be simple --- a type constructed
from non-dependent pair and non-dependent function types over the set variables $s_1$, $\ldots$, $s_n$.
Here we allow $\tau$ to be any (dependent) set expression.

An example of a class expression that is not a structure type is a group-with-action type
$$\sum_{\intype{G\;}{\;\mathrm{Group}}}\;\sum_{\intype{s\;}{\;\sett}}\;S_{\intype{f\;}{\;(\pi_1(G) \rightarrow (s \rightarrow s))}}\;\Phi[G,s,f]$$
where $\Phi[G,s,f]$ states the group action conditions.  Here the first component of a group-with-action is a group rather than a carrier set.

To state the structure rule we will use some additional notation.
For set expressions $s$ and $w$ we let $\mathrm{Bi}(s,w)$ abbreviate the set expression $S_{\intype{f\;}{\;(s\rightarrow w)}}\;\Phi[s,w,f]$ where $\Phi[s,w,f]$ states
that $f$ is a bijection.  Also, for $u$ an instance of the structure type $\sum_{\intype{s_1\;}{\;\sett}} \cdots \sum_{\intype{s_n\;}{\;\sett}}\;\tau$
we define $\dot{u}$ to be the substitution mapping a set variable $s_i$ to the value assigned to that variable in $u$.  For example $\dot{u}(s_1)$ is $\pi_1(u)$ and $\dot{u}(s_2)$
is $\pi_1(\pi_2(u))$.  For a substitution $\dot{u}$ and an expression $e$ we write $\dot{u}(e)$ for the result of applying substitution $\dot{u}$ to expression $e$.
We let $x(u)$ be the element of $\dot{u}(\tau)$ included in $u$. For example, if there are two carrier sets then $x(u)$ is
$\pi_2(\pi_2(u))$. The structure rule is then

\vspace{-3ex}
\unnamed{
  \ant{\sigma = \sum_{\intype{s_1\;}{\;\sett}} \cdots \sum_{\intype{s_n\;}{\;\sett}}\;\tau\;\;\mbox{$\tau$ a set expression}}
  \ant{\Gamma \models \intype{u}{\sigma}}
  \ant{\Gamma \models \intype{v}{\sigma}}
}{
  \ant{\Gamma \models (u =_\sigma v) \Leftrightarrow  \exists\;\; \intype{f_1}{\mathrm{Bi}(\dot{u}(s_1),\dot{v}(s_1))}\cdots\intype{f_n}{\mathrm{Bi}(\dot{u}(s_n),\dot{v}(s_n))}
    \;\;\mathrm{EQ}(\tau,\dot{u},x(u),\dot{v},x(v))}
}

where the formula $\mathrm{EQ}(\tau,\dot{x},x,\dot{y},y)$ is defined below.
In the formula $\mathrm{EQ}(\tau,\dot{x},x,\dot{y},y)$ we have that $\tau$, $x$,and $y$ are expressions and $\dot{x}$ and $\dot{y}$ are
syntactic substitutions with $\Gamma \models \intype{x}{\dot{x}(\tau)}$ and $\Gamma \models \intype{y}{\dot{y}(\tau)}$.
The formula $\mathrm{EQ}(\tau,\dot{x},x,\dot{y},y)$ is defined by the
following clauses.

\begin{eqnarray*}
  \mathrm{EQ}(\tau,\dot{x},x,\dot{y},y) & \equiv & x = y \;\mbox{for
    $\tau$ not containing any $s_i$}
  \\ \\ \mathrm{EQ}(s_i,\dot{x},x,\dot{y},y) & \equiv & f_i(x) = y
  \\ \\ \mathrm{EQ}\left(\left(S_{\intype{z\;}{\;\tau}}
  \Phi\right),\dot{x},x,\dot{y},y\right) & \equiv &
  \mathrm{EQ}(\tau,\dot{x},x,\dot{y},y)
  \\ \\ \mathrm{EQ}\left(\left(\sum_{\intype{z\;}{\;\tau_1}}
  \tau_2\right),\dot{x},x,\dot{y},y\right) & \equiv &
  \left\{\begin{array}{l}\mathrm{EQ}(\tau_1,\dot{x},\pi_1(x),\dot{y},\pi_1(y))
  \\ \\ \wedge \;\;
  \mathrm{EQ}(\tau_2,\dot{x}[z:=\pi_1(x)],\pi_2(x),\dot{y}[z:=\pi_1(y)],\pi_2(y))\end{array}\right. \\ \\ \\ \mathrm{EQ}\left(\left(\prod_{\intype{z\;}{\;\tau_1}}
  \tau_2\right),\dot{f},f,\dot{g},g\right) & \equiv &
  \left\{\begin{array}{l} \forall\;\;
  \intype{x}{\dot{f}(\tau_1)}\;\forall \;\intype{y}{\dot{g}(\tau_1)}
  \\ \\ ~\;\;\; \mathrm{EQ}(\tau_1,\dot{f},x,\dot{g},y)
  \\ \\ ~\;\;\;\;\;\;\Rightarrow
  \mathrm{EQ}(\tau_2,\dot{f}[z:=x],f(x),\dot{g}[z:=y],g(y)) \end{array}\right.
\end{eqnarray*}

\subsection{Functors, Cryptomorphism and Completeness}

Assuming isomorphism congruence, if we have
$\Gamma\;\intype{x}{\sigma} \models \intype{e[x]}{\tau}$ with $x$ not
free in $\tau$ then a lambda expression of the form $(\lambda\;
\intype{x}{\sigma} \;e[x])$ where $\sigma$ is a class expression can
be assigned a semantics as a functor mapping elements of the class
$\sigma$ to elements of the class $\tau$ and respecting isomorphism.
We show in section~\ref{sec:functor} that such a lambda expressions
also defines a mapping of isomorphisms (groupoid morphisms) and hence
defines a functor between groupoids.

Following the terminology of Birkoff \cite{Birk} and Rota \cite{Rota},
two classes $\sigma$ and $\tau$ that are well-formed under $\Gamma$
will be called {\em cryptomorphic} (under $\Gamma$) if there exist
functor-level lambda expressions $\intype{F}{\sigma \rightarrow \tau}$
and $\intype{G}{\tau \rightarrow \sigma}$ that are well-formed under
$\Gamma$ and where we have
$\Gamma;\;\intype{x}{\sigma} \models G(F(x)) = x$ and $\Gamma;\;\intype{y}{\tau} \models F(G(y)) = y$.

For example a group can be defined either as a four-tuple of a set, a
group operation, an inverse operation and an identity element or as a
set and group operation such that inverses and an identity exist.

We now show that every class expression is cryptomorphic to a
structure class.  To determine if $u =_\sigma v$ it then suffices to
show $F(u) =_\tau F(v)$ where $F$ is a cryptomorphism from $\sigma$ to
a structure class $\tau$.  {\bf Expressing isomorphism at structure
  classes then allows us to express isomorphism at every class.}

The proof that every class $\sigma$ is cryptomorphic to a structure
class is by induction on the class expression $\sigma$.  Function
types and the set $\bool$ are set expressions and are by definition
structure classes.  The class $\sett$ is cryptomorphic to
$\sum_{\intype{s\;}{\;\sett}} S_{\intype{P\;}{\;\bool}}\;P$.  For a
subclass expression $S_{\intype{x\;}{\;\sigma}}\;\Phi$ we can assume
by the induction hypothesis that $\sigma$ is cryptomorphic to a
structure class $\sum_{\intype{s_1\;}{\;\sett}} \cdots \sum_{\intype{s_n\;}{\;\sett}}\;\tau$ in which case
$S_{\intype{x\;}{\;\sigma}}\;\Phi[x]$ is cryptomorphic to a class
expression of the form $\sum_{\intype{s_1\;}{\;\sett}} \cdots \sum_{\intype{s_n\;}{\;\sett}}\;S_{\intype{x\;}{\;\tau}}\; \Phi'$.
Finally we consider a dependent pair class
$\sum_{\intype{x\;}{\;\sigma}}\;\tau$.  By the induction hypothesis
$\sigma$ and $\tau$ are cryptomorphic to classes
$$\sum_{\intype{s_1\;}{\;\sett}} \cdots \sum_{\intype{s_n\;}{\;\sett}}\;\tau_1$$ and
$$\sum_{\intype{s_{n+1}\;}{\;\sett}} \cdots \sum_{\intype{s_{n+m}\;}{\;\sett}}\;\tau_2$$ respectively.  In this
case the class $\sum_{\intype{x\;}{\;\sigma}}\;\tau$ is cryptomorphic
to a class of the form
$$\sum_{\intype{s_1\;}{\;\sett}} \cdots \sum_{\intype{s_{n+m}\;}{\;\sett}}\;\sum_{\intype{x\;}{\;\tau_1}}\;\tau_2'$$

\section{A More Rigorous Treatment}
\label{sec:semantics}

This section gives a rigorous treatment of set-theoretic dependent
type theory.  To make the presentation more widely accessible we start
by discussing the universe $V$ of all sets and the set/class
distinction.  We then give a formal recursive definition of Tarskian
semantic value functions which implicitly define well-formedness.
Finally, we address some technical issues arising from the structure
of the recursion defining the value functions.

\subsection{The Universe of Sets}
\label{sec:universe}

We assume a Grothendeick universe $V$.  A Grothendeick universe is a
set whose cardinality is inaccessible (definition omitted) and such
that for $U \subseteq V$ we have $U \in V$ if and only if the
cardinality of $U$ is less than the cardinality of $V$. A Grothendeick
universe is a model of ZFC set theory and by G\"{o}dell's
incompleteness theorem the existence of a Grothendeick universe cannot
be proved from ZFC. However, mathematicians generally accept that we
can talk about the class of all sets, the class of all groups, and so
on.  These classes are subsets of $V$ that are too large to be
elements of $V$.

For conceptual clarity we will work with tagged values. We
define basic values to be hereditarily tagged values over five tags.

\begin{definition}
  \label{def:basic}
  A basic value is either
  \begin{itemize}
  \item a tagged atom $\tuple{\mbox{\tt "atom"},x}$ where $x$ is an
    arbitrary element of $V$,
  \item a tagged truth value $\tuple{\mbox{\tt "Bool"},x}$ were $x$ is
    one of the two values 1 or 0 ($\true$ or $\false$),
  \item a tagged pair $\tuple{\mbox{\tt "pair"}, \tuple{x,y}}$ where
    $x$ and $y$ are (recursively) basic values,
  \item a tagged function $\tuple{\mbox{\tt "function"},f}$ where $f
    \in V$ is a function (a set of pairs) from basic values to basic
    values,
  \item or a tagged set $\tuple{\mbox{\tt "set"},s}$ where $s \in V$
    is a set of basic values.
  \end{itemize}
\end{definition}

This definition is recursive but one can show by induction on
set-theoretic containment (or rank) that there is a unique class
satisfying this recursive specification.

\subsection{Tarskian Value Functions}
\label{sec:tarskian}

We now formally define two semantic value functions. The first assigns
a meaning to well-formed contexts where a context is a sequence of
variable declarations and assumptions.  For a well-formed context
$\Gamma$ we have that ${\cal V}\double{\Gamma}$ is a set of variable
interpretations.  Each variable interpretation is mapping from a
finite set of variables to basic values.  For a variable
interpretation $\gamma \in {\cal V}\double{\Gamma}$, and a well-formed
expression $e$, we define ${\cal V}_\Gamma\double{e}\gamma$ to be the
value (meaning) of $e$ under the values assigned by $\gamma$ to the
variables in $e$.  The meaning functions ${\cal V}\double{\Gamma}$ and
${\cal V}_\Gamma\double{e}\gamma$ are defined by mutual recursion.
The definitions are first presented using simple naive recursion.
This recursion is not just the standard recursion on the structure of
expressions and the technical validity of the recursion is addressed
in section~\ref{sec:index}.

It is important that the value functions are partial.  The value
function ${\cal V}_\Gamma\double{e}\gamma$ is only defined when $e$ is
well-formed under $\Gamma$.  Here we equate well-formedness with
definedness. A context $\Gamma$ is well-formed if and only if ${\cal
  V}\double{\Gamma}$ is defined and an expression $e$ is well-formed
under a well-formed context $\Gamma$ if and only if for all $\gamma
\in {\cal V}\double{\Gamma}$ we have that ${\cal V}\double{e}\gamma$
is defined.

For a variable interpretation $\gamma$, a variable $x$ not assigned a
value in $\gamma$, and a value $v$ we write $\gamma[x:=v]$ for the
variable interpretation $\gamma$ extended to assign value $v$ to
variable $x.$

The semantic value function defined in this section has the property
that well-formed set-level expressions always denote basic values in
$V$ while class expressions denote either the empty class, which we
distinguish from the empty set, or a proper class --- a subset of $V$
too large to be an element of $V$.  We note that the empty class can
be written as $S_{\intype{s\;}{\;\sett}}\;\false$.

The value functions are defined by following recursive clauses.

\begin{itemize}
\item[(V1)] {\color{red} ${\cal V}\double{\epsilon}$}. ${\cal V}\double{\epsilon}$ is defined to be the set containing just the
  empty variable interpretation --- the variable interpretation not
  assigning any value to any variable.

\item[(V2)] {\color{red} ${\cal V}_\Gamma\double{\sett}\gamma$}. If
  ${\cal V}\double{\Gamma}$ is defined and $\gamma \in {\cal
  V}\double{\Gamma}$ then we define ${\cal
  V}_\Gamma\double{\sett}\gamma$ to be the class of all sets.

\item[(V3)] {\color{red} ${\cal V}_\Gamma\double{\bool}\gamma$}. If
  ${\cal V}\double{\Gamma}$ is defined and $\gamma \in {\cal
  V}\double{\Gamma}$ then we define ${\cal
  V}_\Gamma\double{\bool}\gamma$ to be the set containing the two
  truth values $\true$ and $\false$.

\item[(V4)] We say that $u$ is defined under $\Gamma$ if ${\cal  V}\double{\Gamma}$
  is defined and for all $\gamma \in {\cal  V}\double{\Gamma}$ we
  have that ${\cal V}_\Gamma\double{u}\gamma$ is defined.
  
\item[(V5)] We write write $\Gamma \models \inntype{\sigma}{\class}$
  to mean that $\sigma$ is defined under $\Gamma$ and for all
  $\gamma \in {\cal V}\double{\Gamma}$ we have that ${\cal V}_\Gamma\double{\sigma}\gamma$ is a class and similarly write
  $\Gamma \models \inntype{\sigma}{\type}$ to mean that $\sigma$ is
  defined under $\Gamma$ and for all $\gamma \in {\cal V}\double{\Gamma}$ we have that ${\cal
    V}_\Gamma\double{\sigma}\gamma$ is either a set or a class.

\item[(V6)] For a set-level expression $u$ we write $\Gamma \models  \intype{u}{\sigma}$ to mean that $u$ is defined under $\Gamma$ and
  $\Gamma \models \inntype{\sigma}{\type}$ and for all $\gamma \in {\cal V}\double{\Gamma}$
  we have $\tilde{u} \in \tilde{\sigma}$ where $\tilde{u} = {\cal V}_\Gamma\double{u}\gamma$
  and $\tilde{\sigma} = {\cal V}_\Gamma\double{\sigma}\gamma$.

\item[(V7)] For $\Gamma \models \intype{\Phi}{\bool}$ we write
  $\Gamma \models \Phi$ to mean that for all $\gamma \in {\cal V}\double{\Gamma}$ we have ${\cal V}_\Gamma\double{\Phi}\gamma = \true$.

\item[(V8)] {\color{red} ${\cal V}\double{\Gamma;\;\intype{x}{\sigma}}$}. If
  $\Gamma \models \inntype{\sigma}{\type}$, and $x$ is a variable not declared in
  $\Gamma$, then ${\cal V}\double{\Gamma;\;\intype{x}{\sigma}}$ is
  defined to be the class of variable interpretations of the form
  $\gamma[x:=\tilde{u}]$ for $\gamma \in {\cal V}\double{\Gamma}$ and
  $\tilde{u} \in {\cal V}_\Gamma\double{\sigma}\gamma$.

\item[(V9)] {\color{red} ${\cal V}\double{\Gamma;\;\Phi}$}. If $\Gamma \models \intype{\Phi}{\bool}$ then ${\cal V}\double{\Gamma;\;\Phi}$
  is defined to be the class of variable interpretations $\gamma \in  {\cal V}\double{\Gamma}$ such
  that ${\cal V}_\Gamma\double{\Phi}\gamma = \true$.

\item[(V10)] {\color{red} ${\cal V}_\Gamma\double{x}\gamma$}. If $x$
  is declared in $\Gamma$, and $\gamma \in {\cal V}\double{\Gamma}$,
  then we define ${\cal V}_\Gamma\double{x}\gamma$ to be $\gamma(x)$.
  
\item[(V11)] {\color{red} ${\cal V}_\Gamma\double{\sum_{\intype{x\;}{\;\sigma}}\;\tau[x]}\gamma$}. If
  $\Gamma \models \inntype{\sigma}{\type}$ and
  $\Gamma;\intype{x}{\sigma} \models \inntype{\tau[x]}{\type}$ then
  ${\cal V}_\Gamma\double{\sum_{\intype{x\;}{\;\sigma}}\;\tau[x]}\gamma$ is
  defined to be the set or class of all pairs
  $\tuple{\tilde{u},\tilde{v}}$ with $\tilde{u} \in {\cal V}_\Gamma\double{\sigma}\gamma$
  and $\tilde{v} \in {\cal V}_{\Gamma;\intype{x\;}{\;\sigma}}\double{\tau[x]}\gamma[x:=\tilde{u}]$.

\item[(V12)] {\color{red} ${\cal V}_\Gamma\double{S_{\intype{x\;}{\;\sigma}}\;\Phi[x]}\gamma$}. If
  $\Gamma \models \inntype{\sigma}{\type}$ and
  $\Gamma;\intype{x}{\sigma} \models \intype{\Phi[x]}{\bool}$ then
  ${\cal V}_\Gamma\double{S_{\intype{x\;}{\;\sigma}}\;\Phi[x]}\gamma$
  is defined to be the set or class of all $\tilde{u} \in {\cal V}_\Gamma\double{\sigma}\gamma$
  such that ${\cal V}_{\Gamma;\intype{x\;}{\;\sigma}}\double{\Phi[x]}\gamma[x:=\tilde{u}] = \true$.

\item[(V13)] {\color{red} ${\cal V}_\Gamma\double{\prod_{\intype{x\;}{\;\sigma}}\;\tau[x]}\gamma$}. If
  $\Gamma \models \intype{\sigma}{\sett}$ and
  $\Gamma;\intype{x}{\sigma} \models \intype{\tau[x]}{\sett}$ then
  \\ ${\cal V}_\Gamma\double{\prod_{\intype{x\;}{\;\sigma}}\;\tau[x]}\gamma$
  is defined to be the set of all functions $\tilde{f}$ with domain
  ${\cal V}_\Gamma\double{\sigma}\gamma$ and such that for all
  $\tilde{u} \in {\cal V}_\Gamma\double{\sigma}\gamma$ we have
  $\tilde{f}(\tilde{u}) \in \\ {\cal V}_{\Gamma;\intype{x\;}{\;\sigma}}\double{\tau[x]}\gamma[x:=\tilde{u}]$.
  
\item[(V14)] {\color{red} ${\cal V}_\Gamma\double{\Phi \vee
    \Psi}\gamma$}. If $\Gamma \models \intype{\Phi}{\bool}$ and
  $\Gamma \models \intype{\Psi}{\bool}$ then ${\cal
  V}_\Gamma\double{\Phi \vee \Psi}\gamma$ is defined to be $\true$ if
  either ${\cal V}_\Gamma\double{\Phi}\gamma = \true$ or ${\cal
    V}_\Gamma\double{\Psi}\gamma = \true$.  Similar definitions apply
  to the other Boolean operations.
  
\item[(V15)] {\color{red} ${\cal V}_\Gamma\double{f(u)}\gamma$}. If
  $\Gamma \models \intype{f}{\left(\prod_{\intype{x\;}{\;\sigma}}
  \;\tau[x]\right)}$ and $\Gamma \models \intype{u}{\sigma}$ then
  ${\cal V}_\Gamma\double{f(u)}\gamma$ is defined to be
  $\tilde{f}(\tilde{u})$ for $\tilde{f} = {\cal
  V}_\Gamma\double{f}\gamma$ and $\tilde{u} = {\cal
  V}_\Gamma\double{u}\gamma$.
  
\item[(V16)] {\color{red} ${\cal V}_\Gamma\double{u = v}\gamma$}. If
  $\Gamma \models \intype{u}{\sigma}$ and $\Gamma \models
  \intype{v}{\sigma}$ and $\Gamma \models \intype{\sigma}{\sett}$,
  then ${\cal V}_\Gamma\double{u = v}\gamma$ is defined to be $\true$
  if $\tilde{u} = \tilde{v}$ where $\tilde{u} = {\cal
    V}_\Gamma\double{u}\gamma$ and $\tilde{v} = {\cal
    V}_\Gamma\double{v}\gamma$.
  
\item[(V17)] {\color{red} ${\cal
    V}_\Gamma\double{\forall\;{\intype{x}{\sigma}}\;\Phi[x]}\gamma$}. If
  $\Gamma;\;\intype{x}{\sigma} \models \intype{\Phi[x]}{\bool}$ then
  ${\cal
  V}_\Gamma\double{\forall\;{\intype{x}{\sigma}}\;\Phi[x]}\gamma$ is
  defined to be $\true$ if $\Gamma;\intype{x}{\sigma} \models \Phi[x]$
  (and false otherwise).

\item[(V18)] {\color{red} ${\cal
    V}_\Gamma\double{\tuple{u,v}}\gamma$}. If $u$ and $v$ are
  set-level and defined under $\Gamma$ then ${\cal
    V}_\Gamma\double{\tuple{u,v}}\gamma$ is defined to be
  $\tuple{\tilde{u},\tilde{v}}$ with $\tilde{u} ={\cal
    V}_\Gamma\double{u}\gamma$ and $\tilde{v} = {\cal
    V}_\Gamma\double{v}\gamma$.
  
\item[(V19)] {\color{red} ${\cal
    V}_\Gamma\double{\pi_i(u)}\gamma$}. If $\Gamma \models
  \intype{u}{\left(\sum_{\intype{x\;}{\;\sigma}}\;\tau[x]\right)}$
  then ${\cal V}_\Gamma\double{\pi_i(u)}\gamma$ is defined to be
  $\pi_i(\tilde{u})$ where $\tilde{u} = {\cal
    V}_\Gamma\double{u}\gamma$.

\item[(V20)] {\color{red} ${\cal V}_\Gamma\double{\lambda
    \;\intype{x}{\sigma}\;e[x]}\gamma$}. If $\Gamma \models
  \intype{\sigma}{\sett}$ and $\Gamma;\intype{x}{\sigma} \models
  \intype{e[x]}{\tau[x]}$ then \\ ${\cal V}_\Gamma\double{\lambda
    \;\intype{x}{\sigma}\;e[x]}\gamma$ is defined to be the function
  $\tilde{f}$ with domain ${\cal V}_\Gamma\double{\sigma}\gamma$ and
  satisfying $\tilde{f}(\tilde{u}) = {\cal
    V}_{\Gamma;\intype{x\;}{\;\sigma}}\double{e[x]}\gamma[x:=\tilde{u}]$
  for all $\tilde{u} \in {\cal V}_\Gamma\double{\sigma}\gamma$.
\end{itemize}

\subsection{The Recursion Index}
\label{sec:index}

Tarskian value functions are typically defined by structural induction
on expressions.  However, the recursion in clauses (V1) through (V20)
is more subtle.  Clauses (V1) through (V20) can be viewed as acting
like inference rules.  For example consider (V16) defining the meaning
of set-theoretic equality.  Here we must be able to ``derive'' $\Gamma
\models \intype{u}{\sigma}$, $\Gamma \models \intype{v}{\sigma}$ and
$\Gamma \models \intype{\sigma}{\sett}$ in order to derive that the
equality $u=v$ is well-formed.  Once an expression is determined to be
well-formed its value is defined in the usual way by structural
induction on the expression using the semantic values specified in
section~\ref{sec:meaning}.

To formally treat the inference rule nature of clauses (V1) through
(V20) we introduce a recursion index $i \geq 0$ and write ${\cal
  V}^i\double{\Gamma}$ and ${\cal V}_\Gamma^i\double{e}\gamma$.  For
$i = 0$ we take ${\cal V}^0\double{\Gamma}$ and ${\cal
  V}^0_\Gamma\double{e}\gamma$ to be undefined for all $\Gamma$ and
$e$.  We then modify the clauses $(V1)$ through $(V20)$ so that they
define ${\cal V}^{i+1}\double{\Gamma}$ and ${\cal
  V}^{i+1}\double{e}{\gamma}$ in terms of ${\cal V}^i\double{\Gamma}$
and ${\cal V}^i\double{e}{\gamma}$.  This recursion index provides an
induction principle --- induction on the recursion index --- needed in
later proofs.

\begin{itemize}
\item[(V1)] {\color{red} ${\cal V}^{i+1}\double{\epsilon}$}. ${\cal
  V}^{i+1}\double{\epsilon}$ is defined to be the set containing just
  the empty variable interpretation --- the variable interpretation
  not assigning any value to any variable.

\item[(V2)] {\color{red} ${\cal
    V}^{i+1}_\Gamma\double{\sett}\gamma$}. If ${\cal
  V}^i\double{\Gamma}$ is defined and $\gamma \in {\cal
  V}^i\double{\Gamma}$ then we define ${\cal
  V}^{i+1}_\Gamma\double{\sett}\gamma$ to be the class of all sets.

\item[(V3)] {\color{red} ${\cal
    V}^{i+1}_\Gamma\double{\bool}\gamma$}. If ${\cal
  V}^i\double{\Gamma}$ is defined and $\gamma \in {\cal
  V}^i\double{\Gamma}$ then we define ${\cal
  V}^{i+1}_\Gamma\double{\bool}\gamma$ to be the set containing the
  two truth values $\true$ and $\false$.

\item[(V4)] We say that $u$ is defined at i under $\Gamma$ if ${\cal
  V}^i\double{\Gamma}$ is defined and for all $\gamma \in {\cal
  V}^i\double{\Gamma}$ we have that ${\cal
  V}^i_\Gamma\double{u}\gamma$ is defined.
  
\item[(V5)] We write write $\Gamma \models^i \inntype{\sigma}{\class}$
  to mean that $\sigma$ is defined at $i$ under $\Gamma$ and for all
  $\gamma \in {\cal V}^i\double{\Gamma}$ we have that ${\cal
    V}^i_\Gamma\double{\sigma}\gamma$ is a class and similarly write
  $\Gamma \models^i \inntype{\sigma}{\type}$ to mean that $\sigma$ is
  defined at $i$ under $\Gamma$ and for all $\gamma \in {\cal
    V}^i\double{\Gamma}$ we have that ${\cal
    V}^i_\Gamma\double{\sigma}\gamma$ is either a set or a class.

\item[(V6)] For a set-level expression $u$ we write $\Gamma \models^i
  \intype{u}{\sigma}$ to mean that $u$ is defined at $i$ under
  $\Gamma$ and $\Gamma \models^i \inntype{\sigma}{\type}$ and for all
  $\gamma \in {\cal V}^i\double{\Gamma}$ we have $\tilde{u} \in
  \tilde{\sigma}$ where $\tilde{u} = {\cal
    V}^i_\Gamma\double{u}\gamma$ and $\tilde{\sigma} = {\cal
    V}^i_\Gamma\double{\sigma}\gamma$.  We note that $\Gamma \models^i
  \intype{u}{\sigma}$ implies $\tilde{u} \in V$.

\item[(V7)] For $\Gamma \models^i \intype{\Phi}{\bool}$ we write
  $\Gamma \models^i \Phi$ to mean that for all $\gamma \in {\cal
  V}^i\double{\Gamma}$ we have ${\cal V}^i_\Gamma\double{\Phi}\gamma =
  \true$.

\item[(V8)] {\color{red} ${\cal
    V}^{i+1}\double{\Gamma;\;\intype{x}{\sigma}}$}. If $\Gamma
  \models^i \inntype{\sigma}{\type}$, and $x$ is a variable not
  declared in $\Gamma$, then ${\cal
    V}^{i+1}\double{\Gamma;\;\intype{x}{\sigma}}$ is defined to be the
  class of variable interpretations of the form $\gamma[x:=\tilde{u}]$
  for $\gamma \in {\cal V}^i\double{\Gamma}$ and $\tilde{u} \in {\cal
    V}^i_\Gamma\double{\sigma}\gamma$.

\item[(V9)] {\color{red} ${\cal V}^{i+1}\double{\Gamma;\;\Phi}$}. If
  $\Gamma \models^i \intype{\Phi}{\bool}$ then ${\cal
  V}^{i+1}\double{\Gamma;\;\Phi}$ is defined to be the class of
  variable interpretations $\gamma \in {\cal V}^i\double{\Gamma}$ such
  that ${\cal V}^i_\Gamma\double{\Phi}\gamma = \true$.

\item[(V10)] {\color{red} ${\cal V}^{i+1}_\Gamma\double{x}\gamma$}. If
  $x$ is declared in $\Gamma$, and $\gamma \in {\cal
  V}^i\double{\Gamma}$, then we define ${\cal
  V}^{i+1}_\Gamma\double{x}\gamma$ to be $\gamma(x)$.
  
\item[(V11)] {\color{red} ${\cal
    V}^{i+1}_\Gamma\double{\sum_{\intype{x\;}{\;\sigma}}\;\tau[x]}\gamma$}. If
  $\Gamma \models^i \inntype{\sigma}{\type}$ and
  $\Gamma;\intype{x}{\sigma} \models^i \inntype{\tau[x]}{\type}$ then
  ${\cal
  V}^{i+1}_\Gamma\double{\sum_{\intype{x\;}{\;\sigma}}\;\tau[x]}\gamma$
  is defined to be the set or class of all pairs
  $\tuple{\tilde{u},\tilde{v}}$ with $\tilde{u} \in {\cal
    V}^i_\Gamma\double{\sigma}\gamma$ and $\tilde{v} \in {\cal
    V}^i_{\Gamma;\intype{x\;}{\;\sigma}}\double{\tau[x]}\gamma[x:=\tilde{u}]$.

\item[(V12)] {\color{red} ${\cal
    V}^{i+1}_\Gamma\double{S_{\intype{x\;}{\;\sigma}}\;\Phi[x]}\gamma$}. If
  $\Gamma \models^i \inntype{\sigma}{\type}$ and
  $\Gamma;\intype{x}{\sigma} \models^i \intype{\Phi[x]}{\bool}$ then
  ${\cal
  V}^{i+1}_\Gamma\double{S_{\intype{x\;}{\;\sigma}}\;\Phi[x]}\gamma$
  is defined to be the set or class of all $\tilde{u} \in {\cal
    V}^i_\Gamma\double{\sigma}\gamma$ such that ${\cal
    V}^i_{\Gamma;\intype{x\;}{\;\sigma}}\double{\Phi[x]}\gamma[x:=\tilde{u}]
  = \true$.

\item[(V13)] {\color{red} ${\cal
    V}^{i+1}_\Gamma\double{\prod_{\intype{x\;}{\;\sigma}}\;\tau[x]}\gamma$}. If
  $\Gamma \models^i \intype{\sigma}{\sett}$ and
  $\Gamma;\intype{x}{\sigma} \models^i \intype{\tau[x]}{\sett}$ then
  \\ ${\cal
    V}^{i+1}_\Gamma\double{\prod_{\intype{x\;}{\;\sigma}}\;\tau[x]}\gamma$
  is defined to be the set of all functions $\tilde{f}$ with domain
  ${\cal V}^i_\Gamma\double{\sigma}\gamma$ and such that for all
  $\tilde{u} \in {\cal V}^i_\Gamma\double{\sigma}\gamma$ we have
  $\tilde{f}(\tilde{u}) \in \\ {\cal
    V}^i_{\Gamma;\intype{x\;}{\;\sigma}}\double{\tau[x]}\gamma[x:=\tilde{u}]$.
  
\item[(V14)] {\color{red} ${\cal V}^{i+1}_\Gamma\double{\Phi \vee
    \Psi}\gamma$}. If $\Gamma \models^i \intype{\Phi}{\bool}$ and
  $\Gamma \models^i \intype{\Psi}{\bool}$ then ${\cal
  V}^{i+1}_\Gamma\double{\Phi \vee \Psi}\gamma$ is defined to be
  $\true$ if either ${\cal V}^i_\Gamma\double{\Phi}\gamma = \true$ or
  ${\cal V}^i_\Gamma\double{\Psi}\gamma = \true$.  Similar definitions
  apply to the other Boolean operations.
  
\item[(V15)] {\color{red} ${\cal
    V}^{i+1}_\Gamma\double{f(u)}\gamma$}. If $\Gamma \models^i
  \intype{f}{\left(\prod_{\intype{x\;}{\;\sigma}} \;\tau[x]\right)}$
  and $\Gamma \models^i \intype{u}{\sigma}$ then ${\cal
    V}^{i+1}_\Gamma\double{f(u)}\gamma$ is defined to be
  $\tilde{f}(\tilde{u})$ for $\tilde{f} = {\cal
    V}^i_\Gamma\double{f}\gamma$ and $\tilde{u} = {\cal
    V}^i_\Gamma\double{u}\gamma$.
  
\item[(V16)] {\color{red} ${\cal V}^{i+1}_\Gamma\double{u =
    v}\gamma$}. If $\Gamma \models^i \intype{u}{\sigma}$ and $\Gamma
  \models^i \intype{v}{\sigma}$ and $\Gamma \models^i
  \intype{\sigma}{\sett}$, then ${\cal V}^{i+1}_\Gamma\double{u =
    v}\gamma$ is defined to be $\true$ if $\tilde{u} = \tilde{v}$
  where $\tilde{u} = {\cal V}_\Gamma\double{u}\gamma$ and $\tilde{v} =
  {\cal V}^i_\Gamma\double{v}\gamma$.
  
\item[(V17)] {\color{red} ${\cal
    V}^{i+1}_\Gamma\double{\forall\;{\intype{x}{\sigma}}\;\Phi[x]}\gamma$}. If
  $\Gamma;\;\intype{x}{\sigma} \models^i \intype{\Phi[x]}{\bool}$ then
  ${\cal
  V}^{i+1}_\Gamma\double{\forall\;{\intype{x}{\sigma}}\;\Phi[x]}\gamma$
  is defined to be $\true$ if $\Gamma;\intype{x}{\sigma} \models^i
  \Phi[x]$ (and false otherwise).

\item[(V18)] {\color{red} ${\cal
    V}^{i+1}_\Gamma\double{\tuple{u,v}}\gamma$}. If $u$ and $v$ are
  set-level and defined at $i$ under $\Gamma$ then ${\cal
    V}^{i+1}_\Gamma\double{\tuple{u,v}}\gamma$ is defined to be
  $\tuple{\tilde{u},\tilde{v}}$ with $\tilde{u} ={\cal
    V}^i_\Gamma\double{u}\gamma$ and $\tilde{v} = {\cal
    V}^i_\Gamma\double{v}\gamma$.
  
\item[(V19)] {\color{red} ${\cal
    V}^{i+1}_\Gamma\double{\pi_i(u)}\gamma$}. If $\Gamma \models^i
  \intype{u}{\left(\sum_{\intype{x\;}{\;\sigma}}\;\tau[x]\right)}$
  then ${\cal V}^{i+1}_\Gamma\double{\pi_i(u)}\gamma$ is defined to be
  $\pi_i(\tilde{u})$ where $\tilde{u} = {\cal
    V}^i_\Gamma\double{u}\gamma$.

\item[(V20)] {\color{red} ${\cal V}^{i+1}_\Gamma\double{\lambda
    \;\intype{x}{\sigma}\;e[x]}\gamma$}. If $\Gamma \models^i
  \intype{\sigma}{\sett}$ and $\Gamma;\intype{x}{\sigma} \models^i
  \intype{e[x]}{\tau[x]}$ then \\ ${\cal
    V}^{i+1}_\Gamma\double{\lambda \;\intype{x}{\sigma}\;e[x]}\gamma$
  is defined to be the function $\tilde{f}$ with domain ${\cal
    V}^i_\Gamma\double{\sigma}\gamma$ and satisfying
  $\tilde{f}(\tilde{u}) = {\cal
    V}^i_{\Gamma;\intype{x\;}{\;\sigma}}\double{e[x]}\gamma[x:=\tilde{u}]$
  for all $\tilde{u} \in {\cal V}^i_\Gamma\double{\sigma}\gamma$.
\end{itemize}

If $e$ is defined at $i$ for $\Gamma$ then the value of ${\cal
  V}^i_\Gamma\double{e}\gamma$ is defined compositionally on the
structure of the expression $e$ as is normally done in Tarskian
semantics. This implies that for values of $i$ for which ${\cal
  V}^i_\Gamma\double{e}\gamma$ is defined, this value is independent
of the choice of $i$.  This also holds for ${\cal
  V}^i\double{\Gamma}$.  So we can define ${\cal
  V}_\Gamma\double{e}\gamma$ to be ${\cal V}^i\double{\Gamma}\gamma$
for any $i$ where this is defined, and to be undefined if no such $i$
exists. ${\cal V}\double{\Gamma}$ is defined similarly.

\section{Isomorphism and Isomorphism Congruence}
\label{sec:isomorphism}

Section~\ref{sec:bijective} introduces ``bijective values''.  A
bijective value $\tilde{u}$ is one that has both a left projection $L(\tilde{u})$
and a right projection $R(\tilde{u})$.
A bijective set $\tilde{s}$ has the property that the set of pairs $\tuple{L(\tilde{u}),R(\tilde{u})}$
for $\tilde{u} \in \tilde{s}$
defines a bijection between the sets $L(\tilde{s})$ and $R(\tilde{s})$.

To define bijective values we add a new tag to the value space to get values
of the form $\tuple{${\tt "isopair"}$, \tuple{x,y}}$ where $x$ and $y$
are basic values (values that do not contain iso-pairs).  The left
projection of this iso-pair is $x$ and the right projection is $y$.  A
bijective set of iso-pairs is a bijective set.  But sets of pairs and sets of functions
can also be bijective.

Section~\ref{sec:bij} introduces an additional constant symbol $\bij$ into
the formal language denoting the class of all bijective sets. We will call the language including the constant $\bij$ the extended
language while the langauge not including $\bij$ will be called the base language.
The extended language only exists for the purpose of defining isomorphism over all classes
and for proving isomorphism congruence.  Section~\ref{sec:rules} already gives the inference
rules for isomorphism in the base language and a MathZero system can work entirely in the base language.
The extended language is a conservative extension of the base language in the sense that expressions in the base language have the same meaning
in the extended language.

Sections~\ref{sec:bilemma} proves that all values definable in the extended language are bijective. Section~\ref{sec:commutation} proves that the left and right projection functions commute with the
semantic value function.  Section~\ref{sec:functor} proves that functors carry isomorphisms (bijective values) as well as objects.
Section~\ref{subsec:isomorphism} gives the general definition of isomorphism and proves isomorphism congruence.

\subsection{Bijective Values}
\label{sec:bijective}

We first define a class of extended values similar to the class of
basic values of section~\ref{sec:universe} but extended with
iso-pairs.

\begin{definition}
  \label{def;extended}
  An extended value is either
  \begin{itemize}
  \item a tagged atom $\tuple{\mbox{\tt "atom"},x}$ where $x$ is an
    arbitrary element of $V$,
  \item a tagged truth value $\tuple{\mbox{\tt "Bool"},x}$ were $x$ is
    one of the two values 1 or 0 {\normalfont (}$\true$ or
    $\false${\normalfont)},
  \item a tagged pair $\tuple{\mbox{\tt "pair"}, \tuple{x,y}}$ where
    $x$ and $y$ are (recursively) extended values,
  \item a tagged function $\tuple{\mbox{\tt "function"},f}$ where $f$
    is a function (a set of pairs) from extended values to extended
    values,
  \item a tagged set $\tuple{\mbox{\tt "set"},s}$ where $s \in V$ is a
    set of extended values,
  \item or a tagged iso-pair $\tuple{\mbox {\tt "isopair"},
    \tuple{x,y}}$ where $x$ and $y$ are {\bf basic values}.
  \end{itemize}
\end{definition}

\begin{definition}[Left and Right Projections] We define the left projection function $L$ on extended values by the following clauses
  where the right projection function $R$ is defined similarly but with $R(\tuple{\mbox{\tt "isopair"},\tuple{u,w}}) = w.$
  \begin{itemize}
  \item $L(\tuple{\mbox{\tt"atom"},x}) = \tuple{\mbox{\tt "atom"},x}$.
  \item $L(\tuple{\mbox{\tt"Bool"},x}) = \tuple{\mbox{\tt "Bool"},x}$.
  \item $L(\tuple{\mbox{\tt "pair"},\tuple{u,w}}) = \tuple{\mbox{\tt
      "pair"}, \tuple{L(u),L(w)}}$
  \item $L(\tuple{\mbox{\tt "set"},s}) = \tuple{\mbox{\tt "set"},\{L(x):\;x\in s\}}$.
  \item $L(\tuple{\mbox{\tt "function"},f}) = \tuple{\mbox{\tt "function"},g}$ where $g$ is the function containing the
    mappings $L(u) \mapsto L(f(u))$ for $u$ in the domain of $f$.
  \item $L(\tuple{\mbox{\tt "isopair"},\tuple{u,w}}) = u$.
  \end{itemize}
\end{definition}

\begin{definition}
  A {\em bijective value} is either
  \begin{itemize}
  \item a tagged atom $\tuple{\mbox{\tt "atom"},x}$ where $x$ is an
    arbitrary element of $V$,
  \item a tagged truth value $\tuple{\mbox{\tt "Bool"},x}$ were $x$ is
    one of the two values 1 or 0 {\normalfont (}$\true$ or
    $\false${\normalfont)},
  \item a tagged pair $\tuple{\mbox{\tt "pair"}, \tuple{x,y}}$ where
    $x$ and $y$ are (recursively) bijective values,
  \item a tagged function $\tuple{\mbox{\tt "function"},f}$ where $f$
    is a function (a set of pairs) from bijective values to bijective values
    and where the domain of $f$ is a bijective set.
  \item a tagged set $\tuple{\mbox{\tt "set"},s}$ where $s \in V$ is a
    set of bijective values with the property that the set of pairs
    of the form $\tuple{L(\tilde{u}),R(\tilde{u})}$ for $\tilde{u} \in s$
    defines a bijection between $L(s)$ and $R(s)$.
  \item or a tagged iso-pair $\tuple{\mbox {\tt "isopair"}, \tuple{x,y}}$ where $x$ and $y$ are {\bf basic values}.
  \end{itemize}
\end{definition}

For any basic value $\tilde{u}$ we have that $L(\tilde{u}) = R(\tilde{u}) = \tilde{u}$ and $\tilde{u}$ is
bijective.  Also, for any bijection $f$ between two basic sets we
have that the set of iso-pairs of the form $\tuple{\mbox{\tt "isopair"},\tuple{x,f(x)}}$ for $x$ in the domain of $f$ is a
bijective set.
Under the general definition given in section~\ref{subsec:isomorphism} we have that two groups $G$ and $G'$ are isomorphic if there
exists a bijective group $\tilde{G}$ with $L(\tilde{G}) = G$ and $R(\tilde{G}) = G'$.

\subsection{The Constant $\bij$ and the Conservative Extension Lemma}
\label{sec:bij}

We now introduce an additional constant symbol $\bij$ denoting the class of all bijective sets.
This is done by modifying the clauses
of section~\ref{sec:semantics} in two ways.  First we modify the clause for $\sett$ and add a clause for $\bij$ as follows.

\begin{itemize}
\item[(V2)] {\color{red} ${\cal V}_\Gamma\double{\sett}\gamma$}. If
  ${\cal V}\double{\Gamma}$ is defined and $\gamma \in {\cal
  V}\double{\Gamma}$ then we define ${\cal
  V}_\Gamma\double{\sett}\gamma$ to be the class of all {\color{red} basic} sets.

\item[(V2.5)] {\color{red} ${\cal V}_\Gamma\double{\bij}\gamma$}. If
  ${\cal V}\double{\Gamma}$ is defined and $\gamma \in {\cal
  V}\double{\Gamma}$ then we define ${\cal
  V}_\Gamma\double{\bij}\gamma$ to be the class of all {\color{red} bijective} sets.
\end{itemize}

One should keep in mind that we have $\sett \subsetneq \bij$.

We also replace all occurances of $\sett$ in other clauses with $\bij$. For example,

\begin{itemize}
\item[(V16)] {\color{red} ${\cal V}_\Gamma\double{u = v}\gamma$}. If
  $\Gamma \models \intype{u}{\sigma}$ and $\Gamma \models
  \intype{v}{\sigma}$ and $\Gamma \models \intype{\sigma}{{\color{red} \sett}}$,
  then ${\cal V}_\Gamma\double{u = v}\gamma$ is defined to be $\true$
  if $\tilde{u} = \tilde{v}$ where $\tilde{u} = {\cal V}_\Gamma\double{u}\gamma$ and $\tilde{v} = {\cal V}_\Gamma\double{v}\gamma$.
\end{itemize}
becomes
\begin{itemize}
\item[(V16)] {\color{red} ${\cal V}_\Gamma\double{u = v}\gamma$}. If
  $\Gamma \models \intype{u}{\sigma}$ and $\Gamma \models
  \intype{v}{\sigma}$ and $\Gamma \models \intype{\sigma}{{\color{red} \bij}}$,
  then ${\cal V}_\Gamma\double{u = v}\gamma$ is defined to be $\true$
  if $\tilde{u} = \tilde{v}$ where $\tilde{u} = {\cal V}_\Gamma\double{u}\gamma$ and $\tilde{v} = {\cal V}_\Gamma\double{v}\gamma$.
\end{itemize}

The replacement of $\sett$ by $\bij$ also occurs in (V13), the definition of function types, and (V20), the definition of lambda expressions.  The definition of
$\Gamma \models \inntype{\sigma}{\mathbf{Class}}$ and $\Gamma \models \inntype{\sigma}{\mathbf{Type}}$ in clauses (V5) are intpreted as allowing classes and sets of bijective values.
The constant $\sett$ does not occur in other clauses.

\begin{lemma}[Conservative Extension Lemma]
  The extended language is a conservative extension of the base language.  More specifically,
  if ${\cal V}\double{\Gamma}$ is defined in the base langauge then it is also defined in the extended language and the set of variable interpretations
  ${\cal V}\double{\Gamma}$ is the same under both the base and extended languages.  Furthermore,
  if $u$ is defined under $\Gamma$ in the base language then for all $\gamma \in {\cal V}\double{\Gamma}$ we have that
  ${\cal V}_\Gamma\double{u}\gamma$ is defined in the extended language and is equal to the base language value.
\end{lemma}

This lemma is proved by a straightforward induction on the computation index defined in section~\ref{sec:index}.

\subsection{The Bijectivity Lemma}
\label{sec:bilemma}

\begin{lemma}[Bijectivity Lemma]
  If $\Gamma \models \inntype{\sigma}{\class}$ in the extended language then every member of
  ${\cal V}_\Gamma\double{\sigma}\gamma$ is bijective and for every
  set-level expression $u$ defined under $\Gamma$ we have that
  ${\cal V}_\Gamma\double{u}\gamma$ is bijective.
\end{lemma}

\begin{proof}
  We can prove the bijectivity lemma by induction on the recursion
  index $i$ in the clauses of section~\ref{sec:index} modified for the extended language as described above.
  For this we
  note that the lemma is trivial for $i=0$ where nothing is defined.
  We then asume that the lemma holds for $i$ and prove it for $i+1$.
  We need to consider each clause in section~\ref{sec:index} that
  defines a value ${\cal V}^{i+1}_\Gamma\double{e}\gamma$.

  Clauses (V1) and (V4) through (V9) do not define expression values.

  Clauses (V2) and (V3) define the constants $\sett$ and $\bool$. In
  this case the result follows from the fact that elements of these
  types are basic values.  We need to add a clause (V2.5) for the
  constant $\bij$ defining ${\cal V}^{i+1}_\Gamma\double{\bij}\gamma$
  but the lemma holds for $\bij$ since the elements of $\bij$ are
  defined to be bijective sets.

  Clause (V10) defines the variable value ${\cal V}^{i+1}_\Gamma\double{x}\gamma$ to be $\gamma(x)$.  In this case
  we have that $\gamma(x) \in {\cal V}^i_\Gamma\double{\tau}\gamma$
  where $\tau$ is the type expression declared for $x$ in $\Gamma$.
  By the induction hypothesis every element of ${\cal V}^i_\Gamma\double{\tau}\gamma$ is bijective.

  Clause (V11) defines the dependent pair type ${\cal V}^{i+1}_\Gamma\double{\sum_{\intype{x\;}{\;\sigma}}\;\tau[x]}$.
  By the induction hypothesis the bijectivity lemma applies to
  $\sigma$ under $\Gamma$ and to $\tau[x]$ under
  $\Gamma;\intype{x}{\sigma}$.  We then immediately have that every
  element of the dependent pair type is a pair of bijective values
  and every pair of bijective values is bijective.  If $\sigma$
  and $\tau[x]$ are set level we must also show that the pair type is
  bijective.  In this case the induction hypothesis gives us that
  ${\cal V}_\Gamma\double{\sigma}\gamma$ is bijective and for $v \in
  {\cal V}_\Gamma\double{\sigma}\gamma$ we have that ${\cal V}_{\Gamma;\intype{x\;}{\;\sigma}}\double{\tau[x]}\gamma[x:=v]$ is
  bijective. These two statements together imply that ${\cal V}_\Gamma\double{\sum_{\intype{x\;}{\;\tau}}\;\tau[x]}\gamma$ is
  bijective.

  Clause (V12) defines the subtype ${\cal V}^{i+1}_\Gamma\double{S_{\intype{x\;}{\;\sigma}}\;\Phi[x]}$.  In
  this case we have that $\sigma$ is defined at $i$ for $\Gamma$ and
  hence satisfies the induction hypothesis.  This gives that every
  element of the subtype is bijective.  If $\sigma$ is set-level we
  must also show that the subtype is bijective.  But this follows from
  the observation that any subset of a bijective set is bijective.

  Clause (V13) defines the dependent function type ${\cal V}^{i+1}_\Gamma\double{\prod_{\intype{x\;}{\;\sigma}}\;\tau[x]}$.
  Let $\tilde{\sigma}$ denote ${\cal V}^i_\Gamma\double{\sigma}\rho$.
  By the induction hypothesis for $\sigma$ we get that
  $\tilde{\sigma}$ is bijective and hence every function in the
  dependent function type is bijective.  To show that the function
  type is bijective consider a function $f$ in the bijective type and
  consider $\tilde{x} \in \tilde{\sigma}$.  Let $L_{\tilde{\sigma}}$
  be the left operation restricted to $\tilde{\sigma}$.  Let
  $\tilde{\tau}[\tilde{x}]$ denote ${\cal V}^i_{\Gamma\;\intype{x\;}{\;\sigma}}\double{\tau[x]}\gamma[x:=\tilde{x}]$
  and let $L_{\tilde{\tau}[\tilde{x}]}$ be the left operation restricted to $\tilde{\tau}[\tilde{x}]$.
  By the definition of the left operation on functions we have $L(f)(L(\tilde{x}))) = L(f(\tilde{x}))$.  By the definition of the dependent function type we have
  $f(\tilde{x}) \in \tilde{\tau}[\tilde{x}]$ which gives $L(f(\tilde{x})) \in L_{\tilde{\tau}[\tilde{x}]}$.  By the induction hypothesis we have $L_{\tilde{\tau}[\tilde{x}]}$ is bijective
  and we get $f(\tilde{x}) = L_{\tilde{\tau}[\tilde{x}]}^{-1}(L(f)(L_{\tilde{\sigma}}(\tilde{x})))$. So
  the function $f$ can be recovered from its projection $L(f)$ and the function type is bijective.

  Clauses (V14), (V16) and (V17) define formulas --- expressions whose value is Boolean.  Boolean values are basic and hence bijective.

  Clause (V15) defines ${\cal V}^{i+1}_\Gamma\double{f(e)}\gamma$.
  We are given $\Gamma \models^i \intype{f}{\left(\prod_{\intype{x\;}{\;\sigma}}\;\tau[x]\right)}$
  and $\Gamma \models^i \intype{e}{\sigma}$.  This implies that the value of $f(e)$ is in ${\cal V}^i\double{\tau[x]}\gamma[x:=\tilde{e}]$
  where $\tilde{e}$ is the value of $e$.  By the induction hypoothesis for $\tau[x]$ we get that $f(e)$ is bijective.

  Clauses (V18) and (V19) define pairing and the $\pi_1$ and $\pi_2$ projection functions on pairs.
  In each case the bijectivity of the value of the defined expression follows
  immediately from the induction hypothesis.

  Clause (V20) defines ${\cal V}^{i+1}_\Gamma\double{\lambda\;\intype{x}{\sigma}\;e[x]}\gamma$.
  By the induction hypothesis we have that the value of $\sigma$ is bijective
  and hence the value of the lambda expression is bijective.
\end{proof}

\subsection{The Commutation Lemma}
\label{sec:commutation}

\begin{lemma}[Commutation Lemma]
  \label{lem:commute} ~
  The following hold in the extended language.
  \begin{itemize}
  \item[(a)] For $\gamma \in {\cal V}\double{\Gamma}$ we have $L(\gamma) \in {\cal V}\double{\Gamma}$.

  \item[(b)] If $\Gamma \models \inntype{\sigma}{\type}$ then for $\gamma \in {\cal V}\double{\Gamma}$ and $\tilde{u} \in {\cal V}_\Gamma\double{\sigma}\gamma$
    we have $L(\tilde{u}) \in {\cal V}_\Gamma\double{\sigma}L(\gamma)$.
        
  \item[(c)] If $u$ is a set-level expression defined under $\Gamma$ then for $\gamma \in {\cal V}\double{\Gamma}$ we have
    $L({\cal V}_\Gamma\double{u}\gamma) = {\cal V}_\Gamma\double{u}L(\gamma)$
  \end{itemize}
  And (a), (b) and (c) hold similarly for the right projection $R$.
\end{lemma}

\begin{proof}
The proof is by induction on the recursion index.
We must consider each clause that defines either a value for ${\cal V}^{i+1}\double{\Gamma}$
or a value for ${\cal V}^{i+1}_\Gamma\double{e}\gamma$.

All three parts of the commutation lemma are immediate for clauses (V1), (V2) and (V3) defining the constants $\epsilon$, $\sett$ and $\bool$ respectively.
We only note that for (V2) and (V3) and $\gamma \in {\cal V}^i\double{\Gamma}$ the induction hypothesis for $i$ gives $L(\gamma) \in {\cal V}^i\double{\Gamma}$.

Clauses (V4) through (V7) define notation in terms of ${\cal V}^i\double{\Gamma}$ and ${\cal V}^i_\Gamma\double{e}\gamma$ without defining new values.

Clauses (V8) and (V9) define the meaning of nonempty contexts and we must verify part (a) of the lemma.

For (V8) defining ${\cal V}^{i+1}\double{\Gamma; \;\intype{x}{\sigma}}$
consider $\gamma \in {\cal V}^i\double{\Gamma}$ and $\tilde{u} \in {\cal V}^i_\Gamma\double{\sigma}\gamma$.  By part (a) of the induction hypothesis we have $L(\gamma) \in {\cal V}^i\double{\Gamma}$.
By part (b) of the induction hypothesis we have $L(\tilde{u}) \in {\cal V}^i_\Gamma\double{\sigma}L(\gamma)$.  These together give $L(\gamma[x:=\tilde{u}]) \in {\cal V}^{i+1}\double{\Gamma;\;\intype{x}{\sigma}}$.

For (V9) defining ${\cal V}^{i+1}\double{\Gamma;\Phi}$ consider $\gamma \in {\cal V}^i\double{\Gamma}$ such that ${\cal V}^i_\Gamma\double{\Phi}\gamma$ is true.
By part (a) of the induction hypothesis we have $L(\gamma) \in {\cal V}^i\double{\Gamma}$.  By part (c) of the induction hypothesis we have
${\cal V}^i_\Gamma\double{\Phi}\gamma = L({\cal V}^i_\Gamma\double{\Phi}\gamma) = {\cal V}^i_\Gamma\double{\Phi}L(\gamma) = \true$.  This gives $L(\gamma) \in {\cal V}^{i+1}\double{\Gamma;\Phi}$ as desired.

Clauses (V10) through (V20) define ${\cal V}^{i+1}_\Gamma\double{e}\gamma$. In each of these clauses we are assuming $\gamma \in {\cal V}^i\double{\Gamma}$ and by
part (a) of the induction hypothesis we have $L(\gamma) \in {\cal V}^i\double{\Gamma}$.  For each clause we need to show parts (b) and (c).

For (V10) defining the value of a variable condition (b) applies when a set variable $s$ is declared in $\Gamma$ by either $\intype{s}{\sett}$ or $\intype{s}{\bij}$.
In this case condition (b) reduces to $L(\tilde{u}) \in L(\gamma(s))$ for $\tilde{u} \in \gamma(s)$ which follows from the definition of the left operation.  Part (c) is also immediate for this case.

For (V11) defining dependent pair types we first consider part (b). Consider $\tuple{\tilde{u},\tilde{v}} \in {\cal V}^{i+1}_\Gamma\double{\sum_{\intype{x\;}{\;s}}\;\tau[x]}\gamma$.
We must show
$$\tuple{L(\tilde{u}),L(\tilde{v})} \in {\cal V}^{i+1}_\Gamma\double{\sum_{\intype{x\;}{\;\sigma}}\;\tau[x]}L(\gamma).$$
In this case we have $\tilde{u} \in {\cal V}^i_\Gamma\double{\sigma}\gamma$ and $\tilde{v} \in {\cal V}^i_{\Gamma;\;\intype{x}{\;\sigma}}\double{\tau[x]}\gamma[x:=\tilde{u}]$.
The induction hypothesis gives $L(\tilde{u}) \in {\cal V}^i_\Gamma\double{\sigma}L(\gamma)$ and $L(\tilde{v}) \in {\cal V}^i_{\Gamma;\;\intype{x}{\;\sigma}}\double{\tau[x]}L(\gamma)[x:=L(\tilde{u})]$
which proves the result.

Part (c) applies when $\sigma$ and $\tau[x]$ are both set expressions.  In this case we must show
$$L\left({\cal V}^{i+1}_\Gamma\double{\sum_{\intype{x\;}{\;\sigma}}\;\tau[x]}\right)\gamma = {\cal V}^{i+1}_\Gamma\double{\sum_{\intype{x\;}{\;\sigma}}\;\tau[x]}L(\gamma).$$
This follows straightforwardly from the induction hypothesis for part (c) applied to set expressions $\sigma$ and $\tau[x]$.

For (V12) defining subtypes we again first consider part (b). Consider $\tilde{x} \in {\cal V}^{i+1}_\Gamma\double{S_{\intype{x\;}{\;\sigma}}\;\Phi[x]}\gamma$.
We must show
$$L(\tilde{x}) \in {\cal V}^{i+1}_\Gamma\double{S_{\intype{x\;}{\;\sigma}}\;\Phi[x]}L(\gamma).$$
We are given $\tilde{x} \in {\cal V}^i_\Gamma\double{\sigma}\gamma$ and ${\cal V}^i_{\Gamma;\;\intype{x\;}{\;\sigma}}\double{\Phi[x]}\gamma[x:=\tilde{x}] = \true$.
By the induction hypothesis for part (b) we have $L(\tilde{x}) \in {\cal V}^i_\Gamma\double{\sigma}L(\gamma)$ and
by part (c) of the induction hypothesis we have
$$L({\cal V}^i_{\Gamma;\;\intype{x\;}{\;\sigma}}\double{\Phi[x]}L(\gamma)[x:=\tilde{x}]) = {\cal V}^i_{\Gamma;\;\intype{x\;}{\;\sigma}}\double{\Phi[x]}L(\gamma)[x:=L(\tilde{x})] = \true$$
which proves the result.

Part (c) applies when $\sigma$ is a set expression in which case we must show
$$L\left({\cal V}^{i+1}_\Gamma\double{S_{\intype{x\;}{\;\sigma}}\;\Phi[x]}\gamma\right) = {\cal V}^{i+1}_\Gamma\double{S_{\intype{x\;}{\;\sigma}}\;\Phi[x]}L(\gamma)$$
Again this follows straightforwardly from part (c) of the induction hypothesis applied to $\sigma$ and $\Phi[x]$.

For (V13) defining dependent function types we first consider part (b).  For $f \in {\cal V}^{i+1}_\Gamma\double{\prod_{\intype{x\;}{\;\sigma}}\;\tau[x]}\gamma$
we must show
$L(f) \in {\cal V}^{i+1}_\Gamma\double{\prod_{\intype\;{x}{\;\sigma}}\;\tau[x]}L(\gamma)$.
Let $\tilde{\sigma}$ be ${\cal V}^i_\Gamma\double{\sigma}\gamma$.
By the definition of the left operation
the domain of the function $L(f)$ is $L(\tilde{\sigma})$.
By the bijectivity lemma we have that the left operation is a bijection on $\tilde{\sigma}$.  Therefore every element of
$L(\tilde{\sigma})$ can be written as $L(\tilde{x})$ for $\tilde{x} \in \tilde{\sigma}$. So it now suffices
to show that
$$L(f)(L(\tilde{x})) \in {\cal V}^i_{\Gamma;\;\intype{x\;}{\;\sigma}}\double{\tau[x]}L(\gamma)[x:=L(\tilde{x})].$$
By the definition of the left operation on functions we have that $L(f)(L(\tilde{x}))$ equals $L(f(\tilde{x}))$.
Now let $\tilde{\tau}[\tilde{x}]$ denote
${\cal V}^i_{\Gamma;\;\intype{x\;}{\;\sigma}}\double{\tau[x]}\gamma[x:=\tilde{x}]$.
By the definition of dependent function types we have $f(\tilde{x}) \in \tilde{\tau}[\tilde{x}]$.  Part (c) of the induction
hypothesis gives
$L(f(\tilde{x})) \in L(\tilde{\tau}[\tilde{x}]) = {\cal V}^i_{\Gamma;\intype{x\;}{\;\sigma}}\double{\tau[x]}L(\gamma)[x:=L(\tilde{x})]$
which proves part (b).

For dependent function types we have that $\sigma$ and $\tau[x]$ are set expressions and we must show
$$L\left({\cal V}^{i+1}_\Gamma\double{\prod_{\intype{x\;}{\;\sigma}}\;\tau[x]}\gamma\right) = {\cal V}^{i+1}_\Gamma\double{\prod_{\intype{x\;}{\;\sigma}}\;\tau[x]}L(\gamma)$$
Again this follows from the induction hypothesis applied to $\sigma$ and $\tau[x]$.

For the remaining clauses (V14) through (V20) we need only consider part (c).  Clauses (V14) for boolean connectives,  (V15) for function application, (V18) for pairing, and (V19) for $\pi_i$ are
simple compositional definitions and the result follows immediately from the induction hypothesis.  Clauses (V20) is for lambda expressions $\lambda \intype{x}{\sigma}\;e[x]$.
All lambda expressions must be set level and in this case the proof is similarly straightforward by applying the induction hypothesis to $\sigma$ and $e[x]$.
We explicitly consider the two remaining clauses (V16) and (V17).

For (V16) defining equality we must show
$${\cal V}^{i+1}_\Gamma\double{u=v}\gamma = {\cal V}^{i+1}_\Gamma\double{u=v}L(\gamma).$$
In this case we are given that there exists a set expression $\sigma$ such that $\Gamma \models^i \intype{u}{\sigma}$
and $\Gamma \models^i \intype{v}{\sigma}$.  Let $\tilde{u}$, $\tilde{v}$ and $\tilde{\sigma}$
denote ${\cal V}_\gamma\double{u}\gamma$, ${\cal V}_\gamma\double{v}\gamma$ and ${\cal V}_\gamma\double{\sigma}\gamma$ respectively.
By the bijectivity lemma we have that $\tilde{\sigma}$ is bijective.
So we have $\tilde{u} = \tilde{v}$ if and only if
$L(\tilde{u}) = L(\tilde{v})$.  By the induction hypothesis we have $L(\tilde{u}) = {\cal V}^i_\Gamma\double{u}L(\gamma)$ and $L(\tilde{v}) = {\cal V}^i_\Gamma\double{v}L(\gamma)$
which proves the result.

Clause (V17) handles quantified formulas $\forall\;\intype{x}{\sigma}\;\Phi[x]$. Here we have allowed $\sigma$ to be a class expression --- class expressions
within Boolean formulas are allowed. For this case we have
$${\cal V}^{i+1}_\Gamma\double{\forall\;\intype{x}{\sigma}\;\Phi[x]}\gamma = \true \;\;\;\mbox{iff}\;\;\; \Gamma;\intype{x}{\sigma} \models^i \Phi[x].$$
Applying the induction hypothesis to the expression $\Phi[x]$ gives that for $\gamma \in {\cal V}^i\double{\Gamma;\intype{x}{\sigma}}$ we have
${\cal V}^i_{\Gamma;\intype{x\;}{\;\sigma}}\double{\Phi[x]}\gamma = {\cal V}^i_{\Gamma;\intype{x\;}{\;\sigma}}\double{\Phi[x]}L(\gamma)$
from which the result follows.
\end{proof}

\subsection{The Functor Lemma}
\label{sec:functor}

\newcommand{\pack}{\mathrm{pack}}
\newcommand{\unpack}{\mathrm{unpack}}

\begin{definition}
  For a class expression $\sigma$ we define $\sigma^*$ to be the result replacing
  each occurrence of $\sett$ in $\sigma$ with $\bij$.
\end{definition}

\begin{lemma}[Functor Lemma]
  \label{lem:functor}
  If $\Gamma;\intype{x}{\sigma} \models \intype{e[x]}{\tau}$ in the base langauge
  then $\Gamma;\intype{x}{\sigma^*} \models \intype{e[x]}{\tau^*}$ in the extended langauge.
\end{lemma}

The functor lemma can be interpreted as saying that the function mapping
$x \in \sigma$ to $e[x] \in \tau$, in addition to mapping objects in $\sigma$ to objects in $\tau$,
map morphisms in $\sigma^*$ to morphisms in $\tau^*$.

To prove the functor lemma
we introduce two new functions $\pack$ and $\unpack$ on bijective values defined as follows.

\begin{definition}[$\pack$ and $\unpack$]
  ~
  \begin{itemize}
  \item $\pack(\tuple{\mbox{\tt "isopair"},\tuple{u,w}}) = \tuple{\mbox{\tt "atom"},\tuple{\mbox{\tt "isopair"},\tuple{u,w}}}$
  \item $\pack(\tuple{\mbox{\tt "atom"},x}) = \tuple{\mbox{\tt "atom"},\tuple{\mbox{\tt "atom"},x}}$
  \item $\pack(\tuple{\mbox{\tt"Bool"},x}) = \tuple{\mbox{\tt "Bool"},x}$.
  \item $\pack(\tuple{\mbox{\tt "pair"},\tuple{u,w}}) = \tuple{\mbox{\tt "pair"}, \tuple{\pack(u),\pack(w)}}$
  \item $\pack(\tuple{\mbox{\tt "set"},s}) = \tuple{\mbox{\tt "set"},\{\pack(x):\;x\in s\}}$
  \item $\pack(\tuple{\mbox{\tt "function"},f}) = \tuple{\mbox{\tt "function"},g}$ where $g$ is the
    function containing the mappings $\pack(u) \mapsto \pack(f(u))$ for $u$ in the domain of $f$.
  \end{itemize}

  \begin{itemize}
  \item $\unpack(\tuple{\mbox{\tt "atom"},x}) = x$
  \item $\unpack(\tuple{\mbox{\tt"Bool"},x}) = \tuple{\mbox{\tt "Bool"},x}$.
  \item $\unpack(\tuple{\mbox{\tt "pair"},\tuple{u,w}}) = \tuple{\mbox{\tt "pair"}, \tuple{\unpack(u),\unpack(w)}}$
  \item $\unpack(\tuple{\mbox{\tt "set"},s}) = \tuple{\mbox{\tt "set"},\{\unpack(x):\;x\in s\}}$
  \item $\unpack(\tuple{\mbox{\tt "function"},f}) = \tuple{\mbox{\tt "function"},g}$ where $g$ is the
    function containing the mappings $\unpack(u) \mapsto \unpack(f(u))$ for $u$ in the domain of $f$.
  \end{itemize}
\end{definition}

\begin{lemma}
  For any extended value $x$ we have $\unpack(\pack(x)) = x$.
\end{lemma}

\begin{lemma}[Pack Lemma]
  \label{lem:pack}
  The pack operation satisfies the following conditions in the extended langauage
  where we define $\pack(\gamma)$ by $\pack(\gamma)(x) = \pack(\gamma(x))$.

  \begin{itemize}
  \item[(a)] For $\gamma \in {\cal V}\double{\Gamma}$ we have $\pack(\gamma) \in {\cal V}\double{\Gamma}$.
    
  \item[(b)] For any set-level expression $e$ we have that ${\cal V}_\Gamma\double{e}\gamma$ is defined if and only if ${\cal V}_\Gamma\double{e}\pack(\gamma)$ is defined and
  $\pack({\cal V}_\Gamma\double{e}\gamma) = {\cal V}_\Gamma\double{e}\pack(\gamma)$.

  \item[(c)] For ${\cal V}_\Gamma\double{\tau}\gamma$ defined we have that ${\cal V}_\Gamma\double{\tau^*}\gamma$ is also defined
    and for any bijective value $\tilde{u}$ we have $\pack(\tilde{u}) \in {\cal V}_\Gamma\double{\tau}\pack(\gamma)$ if and only if
    $\tilde{u} \in {\cal V}_\Gamma\double{\tau^*}\gamma$.
  \end{itemize}
\end{lemma}

The pack lemma can be proved by induction on the recursion index. We omit the proof.

\begin{proof}[Proof of the Functor Lemma]
We must show that under the conditions of the lemma we have $\Gamma;\;\intype{x}{\sigma^*} \models  \intype{e[x]}{\tau^*}$.
Consider $\gamma \in {\cal V}\double{\Gamma}$ and $\tilde{u} \in {\cal V}_\Gamma\double{\sigma^*}\gamma$.  By condition (a) of the pack lemma
we have $\pack(\gamma) \in {\cal V}\double{\Gamma}$.
By condition (c) of the pack lemma we have $\pack(\tilde{u}) \in {\cal V}_\Gamma\double{\sigma}\pack(\gamma)$.
These together give $\pack(\gamma[x:=\tilde{u}]) \in {\cal V}\double{\Gamma;\intype{x}{\sigma}}$.  We then have
$${\cal V}_\Gamma\double{e[x]}\pack(\gamma[x:=\tilde{u}]) \in {\cal V}_\Gamma\double{\tau}\pack(\gamma[x:=\tilde{u}])$$
By condition (b) of the pack lemma we then have
$$\pack({\cal V}_\Gamma\double{e[x]}\gamma[x:=\tilde{u}]) \in {\cal V}_\Gamma\double{\tau}\pack(\gamma[x:=\tilde{u}])$$
By condition (c) of the pack lemma we have 
$${\cal V}_\Gamma\double{e[x]}\gamma[x:=\tilde{u}] \in {\cal V}_\Gamma\double{\tau^*}\gamma[x:=\tilde{u}]$$
which proves the result.
\end{proof}

\subsection{Isomorphism and Isomorphism Congruence}
\label{subsec:isomorphism}

Part (c) of the pack lemma states that if $\tau$ is a defined class under $\Gamma$ in the base language then $\tau^*$ is defined under $\Gamma$ in the extended language.
We can now give the following definition of isomorphism.

\begin{definition}[Isomorphism]
  If in the base language we have $\Gamma \models \inntype{\tau}{\class}$ and $\Gamma \models \intype{u}{\tau}$ and $\Gamma \models \intype{v}{\tau}$ then for $\gamma \in {\cal V}\double{\Gamma}$
  we define ${\cal V}_\Gamma\double{u =_\tau v}\gamma$
  to be true if there exists $\vec{u} \in {\cal V}_\Gamma\double{\tau^*}\gamma$ with $L(\vec{u}) = {\cal V}_\Gamma\double{u}\gamma$ and $R(\vec{u}) = {\cal V}_\Gamma\double{v}\gamma$.
\end{definition}

\begin{theorem}[Isomorphism Congruence]
  If $\Gamma;\intype{x}{\sigma} \models \intype{e[x]}{\tau}$ where $\Gamma$, $\sigma$, $\tau$ and $e[x]$ do not contain $\bij$, and $x$ does not occur in $\tau$,
  then $\Gamma \models u =_\sigma v$ implies $\Gamma \models e[u] =_\tau e[v]$.
\end{theorem}

\begin{proof}[Proof of Isomorphism Congruence]
  To show $\Gamma \models e[u] =_\tau e[v]$ consider $\gamma \in {\cal V}\double{\Gamma}$.
  Let $\tilde{u}$ be ${\cal V}_\Gamma\double{u}\gamma$ and let $\tilde{v}$ be
  ${\cal V}_\Gamma\double{v}\gamma$. We are given $\Gamma \models u =_\sigma v$ and
  by the definition of $=_\sigma$ we then have that there exists a bijective
  value $\vec{u} \in {\cal V}_\Gamma\double{\sigma^*}\gamma$ with $L(\vec{u}) = \tilde{u}$ and $R(\vec{u}) = \tilde{v}$.  We then have
  $\rho[x:=\vec{u}] \in {\cal V}\double{\Gamma;\;\intype{x}{\sigma^*}}$. Let $\tilde{e}[\tilde{u}]$ be ${\cal V}_{\Gamma;\intype{x\;}{\;\sigma^*}}\double{e[x]}\gamma[x:=\vec{u}]$.
  By the functor lemma we then have $\tilde{e}[\tilde{u}] \in {\cal V}_\Gamma\double{\tau^*}\gamma$.  Since $\Gamma$ does not contain $\bij$ every value in $\gamma$ is basic and
  we have $L(\gamma) = \gamma$. By part (c) of the commutation lemma we then have
  $$L(\tilde{e}[\tilde{u}]) = {\cal V}_{\Gamma;\;\intype{x\;}{\;\sigma^*}}\double{e[x]}\gamma[x:=L(\vec{u})] = {\cal V}_\Gamma\double{e[u]}\gamma$$
  $$R(\tilde{e}[\tilde{u}]) = {\cal V}_{\Gamma;\;\intype{x\;}{\;\sigma^*}}\double{e[x]}\gamma[x:=R(\vec{u})] = {\cal V}_\Gamma\double{e[v]}\gamma$$
  which implies the result.
\end{proof}

\section{Summary}
We have presented well-formedness conditions on purely set-theoretic notation guaranteeing that the expressions built from these notations respect
isomorphism as commonly understood in mathematics.  Isomorphism, symmetry, canonicality, functors, natural transformations and cryptomorphism all emerge
from these well-formedness conditions on set-theoretic language.

\end{document}